\documentclass[a4paper, 11pt]{article}

\usepackage{amsfonts}
\usepackage{graphicx}

\usepackage{amsmath}
\usepackage{amsthm}
\usepackage{amssymb}
\usepackage{enumerate}
\usepackage{fullpage}
\usepackage[hidelinks]{hyperref}

\newtheorem{theorem}{Theorem}
\newtheorem*{theorem*}{Theorem}
\newtheorem{lemma}[theorem]{Lemma}

\newcommand{\TheTitle}{Minimum Perimeter-Sum Partitions in the Plane}

\title{{\TheTitle}\thanks{A preliminary version of this paper appeared at the 33rd International Symposium on Computational Geometry (SoCG 2017).
MA is supported by the Advanced Grant DFF-0602-02499B from the Danish Council for Independent Research under the Sapere Aude research career programme.
MdB, KB, MM, and AM are supported by the Netherlands' Organisation for Scientific Research (NWO) under project no.~024.002.003, 612.001.207, 022.005025, and 612.001.118 respectively.}}

\author{
  Mikkel Abrahamsen\thanks{Department of Computer Science, University of Copenhagen,
\texttt{miab@di.ku.dk}.}
  \and
  Mark de Berg\thanks{Department of Computer Science, TU Eindhoven, the Netherlands,
\texttt{M.T.d.Berg@tue.nl}, \texttt{k.a.buchin@tue.nl}, \texttt{mehran.mehr@gmail.com}, \texttt{admehrabi@gmail.com}.}
  \and
  Kevin Buchin\footnotemark[3]
  \and
  Mehran Mehr\footnotemark[3]
  \and
  Ali D. Mehrabi\footnotemark[3]
}

\usepackage{amsopn}

\renewcommand{\leq}{\leqslant}
\renewcommand{\geq}{\geqslant}

\newcommand{\Reals}{{\Bbb R}}

\newcommand{\bd}{\partial\hspace*{0.5mm}}

\newcommand{\diam}{\mathrm{diam}}

\newcommand{\size}{\mathrm{size}}

\newcommand{\PP}{\ensuremath{\mathcal{P}}}

\newcommand{\eps}{\varepsilon}
\newcommand{\mydef}{:=}
\newcommand{\etal}{et al.}

\newcommand{\ch}{\mbox{\sc ch}}         

\newcommand{\sep}{\mathrm{sep}}
\newcommand{\csep}{c_{\sep}}
\newcommand{\myper}{\mathrm{per}}
\newcommand{\length}{\mathrm{length}}
\newcommand{\mypara}[1]{\medskip\noindent\textbf{#1}}
\newcommand{\base}{\mathrm{base}}
\newcommand{\Sbase}{S_{\base}}
\newcommand{\expand}[1]{\overline{#1}}

\newcommand{\per}[1]{\myper(#1)}
\newcommand{\perPart}[3]{\length(\bd #1(#2,#3))}
\newcommand{\dist}[2]{\lvert#1#2\rvert}
\newcommand{\distPQ}{\mathrm{dist}}
\newcommand{\ellvert}[1]{\ell^{\text{vert}}_{#1}}

\begin{document}

\maketitle

\begin{abstract}
Let $P$ be a set of $n$ points in the plane.
We consider the problem of partitioning $P$ into two subsets $P_1$ and $P_2$ such that the sum of the perimeters of $\ch(P_1)$ and $\ch(P_2)$ is minimized, where $\ch(P_i)$ denotes the convex hull of~$P_i$.
The problem was first studied by Mitchell and Wynters in 1991 who gave an $O(n^2)$ time algorithm.
Despite considerable progress on related problems, no subquadratic time algorithm for this problem was found so far.
We present an exact algorithm solving the problem in $O(n \log^2 n)$ time and a $(1+\eps)$-approximation algorithm running in $O(n + 1/\eps^2\cdot\log^2(1/\eps))$ time.
\end{abstract}



\section{Introduction}
The clustering problem is to partition a given data set into clusters (that is, subsets)
according to some measure of optimality. We are interested in clustering problems where the
data set is a set~$P$ of points in Euclidean space. Most of these clustering problems fall into
one of two categories: problems where the maximum cost of a cluster is given
and the goal is to find a clustering consisting of a minimum number of clusters,
and problems where the number of clusters is given and the goal is to find
a clustering of minimum total cost. In this paper we consider a basic problem of
the latter type, where we wish to find a bipartition~$(P_1,P_2)$ of a planar point set~$P$.
Bipartition problems are not only interesting in their own right,
but also because bipartition algorithms can form the basis of hierarchical clustering methods.

There are many possible variants of the bipartition problem on planar point sets,
which differ in how the cost of a clustering is defined.
A variant that received a lot of attention is the 2-center problem~\cite{c-99,d-84,e-97,jk-94,s-97},
where the cost of a partition~$(P_1,P_2)$ of the given point set~$P$ is
defined as the maximum of the radii of the smallest enclosing disks of~$P_1$ and~$P_2$.
Other cost functions that have been studied include the maximum diameter of the two
point sets~\cite{abky-88}                     
and the sum of the diameters~\cite{h-92};     
see also the survey by Agarwal and Sharir~\cite{as-98} for some more variants.
\medskip

A natural class of cost functions considers the size of the convex hulls $\ch(P_1)$ and $\ch(P_2)$
of the two subsets, where the size of $\ch(P_i)$ can either be defined as the area of~$\ch(P_i)$
or as the perimeter $\myper(P_i)$ of~$\ch(P_i)$. (The perimeter of~$\ch(P_i)$
is the length of the boundary~$\bd{\ch(P_i)}$.) This class of cost functions
was already studied in~1991 by Mitchell and Wynters~\cite{mw-91}.
They studied four problem variants: minimize the sum of the perimeters, the maximum of
the perimeters, the sum of the areas, or the maximum of the areas.
In three of the four variants the convex hulls $\ch(P_1)$ and
$\ch(P_2)$ in an optimal solution may intersect~\cite[full version]{mw-91}---only
in the \emph{minimum perimeter-sum problem} the optimal bipartition is guaranteed
to be a so-called \emph{line partition}, that is, a solution with disjoint convex hulls.
For each of the four variants they gave an $O(n^3)$ algorithm that uses $O(n)$ storage and that computes an optimal line partition; for all except the minimum area-maximum problem they also gave an $O(n^2)$ algorithm that uses $O(n^2)$ storage. Note that (only) for the minimum perimeter-sum problem the computed solution is an optimal bipartition.
Around the same time, the minimum-perimeter sum problem was studied for partitions
into~$k$ subsets for $k>2$; for this variant Capoyleas~\etal~\cite{crw-91}
presented an algorithm with running time $O(n^{6k})$.
Arkin \etal~\cite{akm-93} studied the same problem and gave a similar algorithm.
Very recently, Abrahamsen \etal~\cite{aabcmrrt-18} gave an algorithm for that problem running in time $O(n^{28})$, even when $k$ is part of the input.
Unless $\text{P}=\text{NP}$, this result refutes a conjecture by Arkin \etal~\cite{akm-93} that the problem is $\text{NP}$-complete.

Mitchell and Wynters mentioned the improvement of the space requirement of the quadratic-time
algorithms for the bipartition problems as an open problem, and they stated the existence of a subquadratic
algorithm for any of the four variants as the most prominent open problem.

Rokne~\etal~\cite{rww-92} made progress on the first question, by presenting an
$O(n^2\log n)$ algorithm that uses only~$O(n)$ space for the line-partition version of
each of the four problems. Devillers and Katz~\cite{dk-99} gave algorithms for the min-max variant
of the problem, both for area and perimeter, which run in $O((n+k)\log^2 n)$ time.
Here $k$ is a parameter that is only known to be in $O(n^2)$, although Devillers and Katz
suspected that $k$ is subquadratic. They also gave linear-time algorithms for these problems
when the point set~$P$ is in convex position and given in cyclic order.
Segal~\cite{s-02} proved an $\Omega(n\log n)$ lower bound for the min-max problems.
Very recently, and apparently unaware of some of the earlier work on these problems,
Bae~\etal~\cite{bcess-2016} presented an $O(n^2\log n)$ time algorithm for the
minimum-perimeter-sum problem and an $O(n^4\log n)$ time algorithm
for the minimum-area-sum problem (considering all partitions, not only line partitions).
Despite these efforts, the main question is still open: is it possible to obtain
a subquadratic algorithm for any of the four bipartition problems based
on convex-hull size?

\subsection{Our contribution}
We answer the question above affirmatively by presenting a subquadratic algorithm
for the minimum perimeter-sum bipartition problem in the plane.

As mentioned, an optimal solution $(P_1,P_2)$ to the minimum perimeter-sum
bipartition problem must be a line partition.
A straightforward algorithm would generate all $\Theta(n^2)$ line partitions
and compute the value $\myper(P_1)+\myper(P_2)$ for each of them. If the latter
is done from scratch for each partition, the resulting algorithm runs
in~$O(n^3 \log n)$ time. The algorithms by Mitchell and Wynters~\cite{mw-91}
and Rokne~\etal~\cite{rww-92} improve on this by using the fact that the
different line bipartitions can be generated in an ordered way, so that subsequent
line partitions differ in at most one point. Thus the convex hulls do not have to be recomputed
from scratch, but they can be obtained by updating the convex hulls of the previous bipartition.
To obtain a subquadratic algorithm a fundamentally new approach is necessary:
we need a strategy that generates a subquadratic number of candidate partitions,
instead of considering all line partitions. We achieve this as follows.

We start by proving that an optimal bipartition~$(P_1,P_2)$ has the following
property: either there is a set of~$O(1)$ canonical orientations such that
$P_1$ can be separated from $P_2$ by a line with a canonical orientation,
or the distance between $\ch(P_1)$ and $\ch(P_2)$ is~$\Omega(\min(\myper(P_1),\myper(P_2)))$.
There are only $O(n)$ 
bipartitions of the former type, and finding the
best among them is relatively easy. The bipartitions of the second type are
much more challenging. We show how to employ a compressed quadtree to generate a
collection of $O(n)$ canonical 5-gons---intersections of axis-parallel rectangles and
canonical halfplanes---such that the smaller of $\ch(P_1)$ and $\ch(P_2)$
(in a bipartition of the second type) is contained in one of the 5-gons.

Even though the number of such bipartitions is linear, we cannot afford
to compute their perimeters from scratch.
We therefore use the data structure of Oh and Ahn~\cite{oa-18} to quickly compute $\myper(P\cap Q)$, where $Q$ is a query canonical 5-gon.
Given a set $\mathcal O$ of $k$ orientations, Oh and Ahn described how to create a data structure using $O(nk^3\log^2 n)$ time and space to answer queries of the following type in time $O(k\log^2 n)$:
Given a convex polygon $Q$ where each edge has an orientation in $\mathcal O$, what is $\myper(P\cap Q)$?
In our case, each query polygon $Q$ is the intersection of an axis-parallel square and a canonical halfplane bounded by a line with one of $C=O(1)$ different orientations.
We therefore make $C$ different instances of the data structure, where each instance has as orientations $\mathcal O$ the two axis-parallel directions and one of the $C$ different orientations of the canonical halfplanes (i.e., $k=3$).\footnote{In a preliminary version of this paper~\cite{adkmm-17}, we described a less efficient data structure answering these queries in time $O(\log^4 n)$, resulting in the total running time $O(n\log^4 n)$.
After that Oh and Ahn~\cite{oa-18} developed a more efficient data structure that, as they already observed, can be used to speed up our algorithm.}

To sum up, our main result is an exact algorithm for the minimum perimeter-sum bipartition problem that runs in $O(n\log^2 n)$~time.
As our model of computation we use the real RAM (with the capability of taking square roots) so that we can compute the exact perimeter of a convex polygon---this is necessary to compare the costs of two competing clusterings. We furthermore make the (standard) assumption that the model of computation allows us to compute a compressed quadtree of $n$ points in $O(n\log n)$ time; see footnote~3 in Section~\ref{sec:dist}.

Besides our exact algorithm, we present a linear-time $(1+\eps)$-approximation algorithm.
Its running time is $O(n+T(1/\eps^2))=O(n + 1/\eps^2\cdot \log^2 (1/\eps))$, where $T(1/\eps^2)$ is the running time of an exact algorithm on an instance of size $1/\eps^2$.

\section{The exact algorithm}\label{se:exact-alg}
In this section we present an exact algorithm for the minimum-perimeter-sum partition
problem.  We first prove a separation property that an optimal solution must
satisfy, and then we show how to use this property to develop a fast algorithm.

Let $P$ be the set of $n$ points in the plane for which we want to solve the
minimum-perimeter-sum partition problem. An optimal
partition $(P_1,P_2)$ of~$P$ has the following two basic properties:
$P_1$ and $P_2$ are non-empty, and the convex hulls $\ch(P_1)$ and $\ch(P_2)$
are disjoint~\cite[full version]{mw-91}. In the remainder, whenever we talk about a partition of $P$,
we refer to a partition with these two properties.

\subsection{Geometric properties of an optimal partition}\label{subse:geometric}
Consider a partition $(P_1,P_2)$ of~$P$. Define
$\PP_1\mydef\ch(P_1)$ and $\PP_2\mydef\ch(P_2)$ to be the convex hulls of $P_1$
and $P_2$, respectively, and let $\ell_1$ and $\ell_2$ be the two inner common tangents of~$\PP_1$ and~$\PP_2$.
The lines $\ell_1$ and $\ell_2$ define four wedges: one containing $P_1$,
one containing~$P_2$, and two empty wedges. We call the opening angle of the
empty wedges the \emph{separation angle} of $P_1$ and $P_2$. Furthermore, we
call the distance between $\PP_1$ and $\PP_2$ the
\emph{separation distance} of $P_1$ and~$P_2$.
%
\begin{theorem}\label{th:separation-property}
Let $P$ be a set of $n$ points in the plane, and let $(P_1,P_2)$ be a partition of $P$
that minimizes $\myper(P_1) + \myper(P_2)$. Then the separation angle of $P_1$ and $P_2$
is at least $\pi/6$ or the separation distance is at least
$\csep \cdot \min(\myper(P_1),\myper(P_2))$, where $\csep\mydef 1/250$.
\end{theorem}
%
The remainder of this section is devoted to proving Theorem~\ref{th:separation-property}.
To this end let $(P_1,P_2)$ be a partition of $P$
that minimizes $\myper(P_1) + \myper(P_2)$.
Let $\ell_3$ and $\ell_4$ be the outer common tangents of
$\PP_1$ and $\PP_2$. We define~$\alpha$ to be the angle between $\ell_3$
and $\ell_4$. More precisely, if $\ell_3$ and $\ell_4$ are parallel
we define $\alpha \mydef 0$, otherwise we define $\alpha$ as the opening angle
of the wedge defined by $\ell_3$ and $\ell_4$ containing $\PP_1$ and $\PP_2$.
We denote the separation angle of $P_1$ and $P_2$ by~$\beta$; see Fig.~\ref{fig:exact}.
\begin{figure}
\begin{center}
\includegraphics{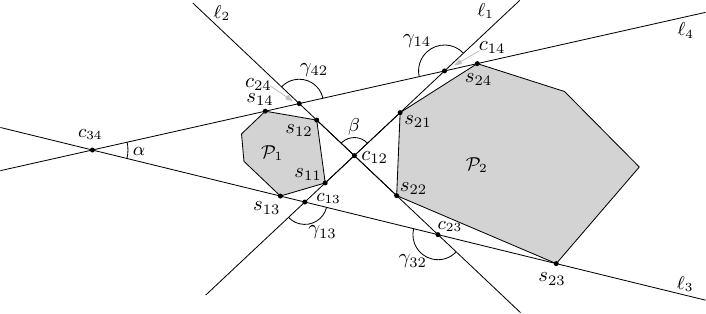}
\end{center}
\caption{The setup in the proof of Theorem~\ref{th:separation-property}.}
\label{fig:exact}
\end{figure}

The idea of the proof is as follows. Suppose that the separation
distance and the separation angle~$\beta$ are both relatively small. Then the region $A$
in between $\PP_1$ and $\PP_2$  and bounded from the bottom by $\ell_3$ and
from the top by $\ell_4$ is relatively narrow.
But then the left and right parts of $\bd A$ (which are contained in $\bd \PP_1$ and $\bd\PP_2$) would be longer than the bottom and
top parts of $\bd A$ (which are contained in $\ell_3$ and $\ell_4$),
thus contradicting the assumption that $(P_1,P_2)$ is an optimal partition.
To make this idea precise, we first prove that if the separation angle~$\beta$ is small, then the angle $\alpha$
between $\ell_3$ and $\ell_4$ must be large. Second, we show that there is a value $f(\alpha)$ such that
the distance between $\PP_1$ and $\PP_2$ is at least $f(\alpha)\cdot \min(\myper(P_1),\myper(P_2))$.
Finally we argue that this implies that if the separation angle is smaller than $\pi/6$, then (to avoid the
contradiction mentioned above) the separation distance must be relatively large. 
Next we present our proof in detail.

Let $c_{ij}$ be the intersection point between $\ell_i$ and $\ell_j$, where~$i<j$.
If $\ell_3$ and $\ell_4$ are parallel,
we choose $c_{34}$ as a point at infinity on $\ell_3$.
Assume without loss of generality that neither $\ell_1$ nor $\ell_2$
separate $\PP_1$ from $c_{34}$, and that
$\ell_3$ is the outer common tangent such that $\PP_1$ and $\PP_2$ are
to the left of $\ell_3$ when traversing $\ell_3$
from $c_{34}$ to an intersection
point in $\ell_3\cap \PP_1$.
Assume furthermore that
$c_{13}$ is closer to $c_{34}$ than $c_{23}$.

For two lines, rays, or segments
$r_1,r_2$, let $\angle(r_1,r_2)$ be the angle we need to rotate $r_1$ in
a counterclockwise direction until $r_1$ and $r_2$ are parallel.
For three points $a,b,c$, let $\angle (a,b,c)\mydef\angle (ba,bc)$.
For $i=1,2$ and $j=1,2,3,4$, let $s_{ij}$ be a point in $P_i\cap \ell_j$.
Let $\bd\PP_i$ denote the boundary of $\PP_i$ and
$\per{\PP_i}$ the perimeter of $\PP_i$. Furthermore, let $\bd\PP_i(x,y)$ denote the
portion of $\bd\PP_i$ from $x\in\bd \PP_i$ counterclockwise to $y\in\bd \PP_i$,
and $\perPart{\PP_i}xy$ denote the length of $\bd\PP_i(x,y)$.

\begin{lemma}\label{movingPoint}
Let $p_0$ and $q$ be points and $\mathbf v$ be a unit vector.
Let $p(t)\mydef p_0+t\cdot\mathbf v$ and $d(t)\mydef \dist{p(t)}{q}$
and assume that $p(t)\neq q$ for all $t\in\Reals$.
Then $d'(t)=\cos(\angle(q,p(t),p(t)+\mathbf v))$ if the points
$q,p(t),p(t)+\mathbf v$ make a left-turn
and $d'(t)=-\cos(\angle(q,p(t),p(t)+\mathbf v))$ otherwise.\footnote{Note
that $\angle(q,p(t),p(t)+\mathbf v)=\angle(q,p(t),p(t)-\mathbf v)$
by the definition of $\angle(\cdot,\cdot,\cdot)$ which is the reason
that there are two cases in the lemma.}
\end{lemma}
\begin{proof}
We prove the lemma for an arbitrary
value $t=t_0$.
By reparameterizing $p$, we may assume that $t_0=0$.
Furthermore, by changing the coordinate system, we can
without loss of generality assume that
$p_0=(0,0)$ and $q=(x,0)$ for some value $x>0$.

Let $\phi\mydef\angle((x,0),(0,0),\mathbf v)$.
Assume that $\mathbf v$ has positive $y$-coordinate---the
case that $\mathbf v$ has negative $y$-coordinate can be handled
analogously.
We have proved the lemma if we manage to show that $d'(0)=-\cos\phi$.
Note that since $\mathbf v$ has positive $y$-coordinate,
we have $p(t)=(t\cos \phi,t\sin\phi)$ for every $t\in\Reals$.
Hence
\[d(t)\ =\ \sqrt { \left( t\cos  \phi  -x \right) ^{2}+{t}^{2}
\sin^2 \phi}
\]
and
\[d'(t)\ =\ \frac {t-x\cos \phi }{\sqrt {{t}^{2}-2tx\cos \phi
+{x}^{2}}}.
\]
 Evaluating at $t=0$, we get
 \[d'(0)\ =\ -{\frac {x\cos \phi }{|x|}}=-\cos \phi,\]
 where the last equality follows since $x>0$.
\end{proof}

\begin{lemma}\label{bigAngles} We have
$\alpha+3\beta\geq\pi$.
\end{lemma}
\begin{proof}
Since $\per {\PP_1}+\per {\PP_2}$ is minimum, we know that
\[
    \perPart {\PP_1}{s_{13}}{s_{14}}+\perPart {\PP_2}{s_{24}}{s_{23}}
    \ \leq\  \Psi,
\]
where
$\Psi\mydef
\dist{s_{13}}{s_{23}}+\dist{s_{14}}{s_{24}}$.
Furthermore, we know that
$s_{11},s_{12}\in\bd \PP_1(s_{13},s_{14})$ and
$s_{21},s_{22}\in\bd \PP_1(s_{24},s_{23})$. We thus have
\[
    \perPart {\PP_1}{s_{13}}{s_{14}}+\perPart{\PP_2}{s_{24}}{s_{23}}
    \ \geq\ \Phi,
\]
where
$
\Phi\mydef
\dist{s_{13}}{s_{11}}+\dist{s_{11}}{s_{12}}+\dist{s_{12}}{s_{14}}+
\dist{s_{24}}{s_{21}}+\dist{s_{21}}{s_{22}}+\dist{s_{22}}{s_{23}}$.
Hence, we must have
\begin{equation}\label{ineq}
\Phi\ \leq\ \Psi.
\end{equation}
Now assume that $\alpha+3\beta<\pi$. We will show that this assumption,
together with inequality~\eqref{ineq}, leads to a contradiction, thus proving the lemma.
To this end we will argue that if~\eqref{ineq} holds, then there exist points $s'_{ij}$ for $i=1,2$ and $j=1,2,3,4$, where $s'_{ij}$ is a point on $\ell_j$, with the following proporties:
\begin{enumerate}[(i)]
\item $\Phi'\leq\Psi'$,
where $\Phi'$ and $\Psi'$ are defined as $\Phi$ and $\Psi$ when each point $s_{ij}$ is replaced by $s'_{ij}$,\label{prop:iii}
\item $s'_{21}$ or $s'_{22}$ coincides with $c_{12}$, and\label{prop:i}
\item $s'_{11}$ or $s'_{12}$ coincides with $c_{12}$.\label{prop:ii}
\end{enumerate}
To finish the proof it then suffices to observe that properties \eqref{prop:iii}--\eqref{prop:ii} together contradict the triangle inequality.

Note that the point $s'_{ij}$ is not required to be contained in $P_i$.
In particular, the points $s'_{13}$ and $s'_{14}$ will in some cases be on the other side of $c_{34}$ than the points $s_{13}$ and $s_{14}$.
In that case there is no pair of convex polygons with outer common tangents defined by $(s'_{13},s'_{23})$ and $(s'_{14},s'_{24})$.
The contradiction applies to distances between a configuration of points that need not be realizable as the supporting points of the common tangents of two convex polygons.
\medskip

To prove the existence of the points $s'_{ij}$ with the claimed properties, we initially define $s'_{ij}\mydef s_{ij}$, so that property~\eqref{prop:iii} is satisfied.
Then we will move the points $s'_{ij}$ (where each $s'_{ij}$ moves on $\ell_{j}$) so that property~\eqref{prop:iii} is preserved throughout the movements and properties~\eqref{prop:i} and~\eqref{prop:ii} are satisfied at the end of the movements.

We first show how to create a situation where \eqref{prop:i} holds, and \eqref{prop:iii} still holds as well.
Let $\gamma_{ij}\mydef \angle(\ell_i,\ell_j)$. We consider two cases.
\begin{itemize}
\item \emph{Case~(A): $\gamma_{32}<\pi-\beta$.} \\[2mm]
    We observe that moving $s'_{23}$ along $\ell_3$ away from $s'_{13}$ increases $\Psi'$ more than it increases $\Phi'$, so property~\eqref{prop:iii} is preserved by such a movement.
    Note that $\angle (xs'_{23},\ell_2)\geq \gamma_{32}$ for any $x\in s'_{22}c_{12}$.
    However, by moving $s'_{23}$ sufficiently far away we can make $\angle (xs'_{23},\ell_2)$ arbitrarily close to $\gamma_{32}$.
    We therefore move $s'_{23}$ so far away that $\angle (xs'_{23},\ell_2) <\pi-\beta$ for any point~$x\in s'_{22}c_{12}$.
    We now consider what happens as we let a point $x$ move at unit speed from $s'_{22}$ towards $c_{12}$.
    To be more precise, let $T\mydef\dist{s'_{22}}{c_{12}}$, let
    $\mathbf v$ be the unit vector with direction from $c_{23}$ to $c_{12}$,
    and for any $t\in[0,T]$ define $x(t)\mydef s'_{22}+t\cdot \mathbf v$.
    Note that $x(0)=s'_{22}$ and $x(T)=c_{12}$.

    Let $a(t)\mydef \dist{x(t)}{s'_{23}}$ and
    $b(t)\mydef \dist{x(t)}{s'_{21}}$.
    Lemma \ref{movingPoint} gives that
    \[
    a'(t)=-\cos(\angle(x(t)s'_{23},\ell_2)) \quad \mbox{and}\quad 
    b'(t)=\cos(\angle(\ell_2,x(t)s'_{21})).
    \]
    Since $\angle (x(t)s'_{23},\ell_2) <\pi-\beta$ for any value
    $t\in[0,T]$, we get $a'(t)< -\cos(\pi-\beta)$.
    Furthermore, we have $\angle(\ell_2,x(t)s'_{21})\geq \pi-\beta$ and hence
    $b'(t)\leq \cos(\pi-\beta)$.
    Therefore, $a'(t)+b'(t)< 0$ for any $t$ and we conclude that
    $a(T)+b(T)\leq a(0)+b(0)$. This is the same as
    $\dist{s'_{21}}{c_{12}}+\dist{c_{12}}{s'_{23}}\leq
    \dist{s'_{21}}{s'_{22}}+\dist{s'_{22}}{s'_{23}}$,
    so we now move $s'_{22}$ to~$c_{12}$ and are ensured that~\eqref{prop:iii} still holds.
\item \emph{Case~(B): $\gamma_{32}\geq\pi-\beta$.} \\[2mm]
    Using our assumption $\alpha+3\beta<\pi$ we get $\gamma_{32}> \alpha+2\beta$.
    Note that $\gamma_{14}=\pi-\gamma_{32}+\alpha+\beta$.
    Hence, $\gamma_{14}<\pi-\beta$.
    By first moving $s'_{24}$ away from $s'_{14}$ and then $s'_{21}$ towards $c_{12}$, we can argue, similarly to Case~(A), that we can reach a situation where~\eqref{prop:iii} still holds and $s'_{21}$ coincides with $c_{12}$.
\end{itemize}
We conclude that in both cases we can ensure~\eqref{prop:i} without violating~\eqref{prop:iii}.

Since $\gamma_{13}\leq\gamma_{14}$ and $\gamma_{42}\leq \gamma_{32}$,
we likewise have $\gamma_{13}<\pi-\beta$ or $\gamma_{42}<\pi-\beta$.
Hence, by first moving $s'_{13}$ or $s'_{14}$ and since then $s'_{11}$ or $s'_{12}$, we can in a similar way reach a situation where $s'_{11}$ or $s'_{12}$ coincides with $c_{12}$ without violating~\eqref{prop:iii}, thus ensuring~\eqref{prop:ii} and finishing the proof.
\end{proof}

The following technical lemma is illustrated in Fig.~\ref{fig:quadrilateral}.
The lemma will be used in the proof of the subsequent Lemma~\ref{distanceLemma}.
The overall idea in the two lemmata is that we consider pushing $\PP_2$ towards $\PP_1$ until they touch.
In the configuration where they touch, $m$ in Lemma~\ref{simpleDiff} corresponds to a common point, $r_1,r_2$ correspond to the outer common tangents, and $b_1,t_1$, resp.~$b_2,t_2$, correspond to the points where $\PP_1$, resp.~$\PP_2$, supports $r_1,r_2$.
The lemma then gives a lower bound on how much cheaper it would be to unite $\PP_1$ and $\PP_2$.
This in turn implies a lower bound on how far we pushed $\PP_2$ (using that $(P_1,P_2)$ was assumed to be an optimal bipartition), which is a lower bound on the original distance between $\PP_1$ and $\PP_2$, as stated in Theorem~\ref{th:separation-property}.
\begin{figure}
\begin{center}
\includegraphics{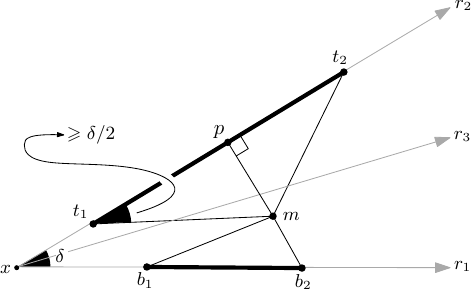}
\end{center}
\caption{Illustration for Lemma~\protect\ref{simpleDiff}. $\Phi$ is the total length of the
four segments $t_1m$, $t_2m$, $b_1m$, $b_2m$, and
$\Psi$ is the total length of the two fat segments.}
\label{fig:quadrilateral}
\end{figure}
\begin{lemma}\label{simpleDiff}
Let $x$ be a point and $r_1$ and $r_2$ be two rays starting at $x$ such that
$\angle (r_1,r_2)=\delta$, and assume
that $\delta\leq\pi$.
Let $b_1,b_2\in r_1$ and
$t_1,t_2\in r_2$ be such that $b_1\in xb_2$ and $t_1\in xt_2$, and let
$m$ be a point in the wedge bounded by $r_1$ and $r_2$.
Then
\[
    \Phi-\Psi
    \ \geq\ \frac{(1-\cos(\delta/2)) \cdot \sin (\delta/2)}{1+\sin(\delta/2)}\cdot
(\dist{b_1}{m}+\dist{t_1}{m}),
\]
where
$\Phi\mydef\dist{b_1}{m}+\dist{t_1}{m}+\dist{b_2}{m}+\dist{t_2}{m}$ and
$\Psi\mydef\dist{b_1}{b_2}+\dist{t_1}{t_2}$.
\end{lemma}
\begin{proof}
First note that

\begin{equation}\label{ineq1a}
\dist{b_1}{m}+\dist{b_2}{m}\ \geq\ \dist{b_1}{b_2}
\end{equation}
and
\begin{equation}\label{ineq1b}
\dist{t_1}{m}+\dist{t_2}{m}\ \geq\ \dist{t_1}{t_2}.
\end{equation}

Let $r_3$ be the angular bisector of $r_1$ and $r_2$. Assume without loss of
generality that $m$ lies in the wedge defined by~$r_1$ and~$r_3$.
Then $\angle (m,t_1,t_2)\geq\delta/2$.

We now consider two cases.
\begin{itemize}
\item \emph{Case~(A): $\dist{t_1}{m}\geq \frac{\sin (\delta/2)}{1+\sin(\delta/2)}\cdot(\dist{b_1}{m}+\dist{t_1}{m})$.}
    \\[2mm]
Our first step is to prove that
\begin{eqnarray}
    \dist{t_1}{m}+\dist{t_2}{m}-\dist{t_1}{t_2}
    \geq (1-\cos (\delta/2))\cdot\dist{t_1}{m}.  \label{eq:step2}
\end{eqnarray}
Let $p$ be the orthogonal projection of
$m$ on $r_2$. Note that
$\dist{t_2}m\ \geq\ \dist{t_2}p$.
Consider first the case that
$p$ is on the same side of $t_1$ as $x$. In this case
$\dist{t_2}{p}\geq \dist{t_1}{t_2}$ and therefore
\begin{eqnarray*}
    \dist{t_1}{m}+\dist{t_2}{m}-\dist{t_1}{t_2}
    \geq \dist{t_1}{m}\geq (1-\cos (\delta/2))\cdot\dist{t_1}{m},
\end{eqnarray*}
which proves \eqref{eq:step2}.

Assume now that $p$ is on the same side of $t_1$ as $t_2$.
In this case, we have $\angle(m,t_1,t_2)\leq\pi/2$ and thus
$\dist{t_1}p=\cos (\angle(m,t_1,t_2))\cdot\dist{t_1}m
\leq \cos (\delta/2)\cdot \dist{t_1}m$.
Hence we have
\begin{eqnarray*}
    \dist{t_1}{m}+\dist{t_2}{m}-\dist{t_1}{t_2}
        & \geq & \dist{t_1}{m}+\dist{t_2}{p}-(\dist{t_1}{p}+\dist{t_2}{p}) \\
        & \geq & (1-\cos (\delta/2))\cdot\dist{t_1}{m},
\end{eqnarray*}
and we have proved \eqref{eq:step2}.

    We now have
    \[
    \begin{array}{lll}
    \Phi-\Psi & = & \dist{b_1}{m}+\dist{t_1}{m}+\dist{b_2}{m}+\dist{t_2}{m} - \dist{b_1}{b_2} - \dist{t_1}{t_2} \\
              & \geq & \dist{b_1}{m}+\dist{b_2}{m} - \dist{b_1}{b_2} + (1-\cos (\delta/2))\cdot\dist{t_1}{m}\quad\text{by \eqref{eq:step2}} \\
              & \geq & (1-\cos (\delta/2))\cdot \frac{\sin (\delta/2)}{1+\sin(\delta/2)}\cdot(\dist{b_1}{m}+\dist{t_1}{m})\quad\text{by \eqref{ineq1a}} \\
    \end{array}
    \]
    where the last step uses that we are in Case~(A). Thus the lemma holds in Case~(A).
\item \emph{Case~(B): $\dist{t_1}{m} < \frac{\sin (\delta/2)}{1+\sin(\delta/2)}\cdot(\dist{b_1}{m}+\dist{t_1}{m})$.}
   \\[2mm]
    The condition for this case can be rewritten as
    \begin{equation}\label{ineq2}
    \dist{b_1}{m} \ >\ \frac{1}{1+\sin\delta/2}\cdot (\dist{b_1}{m}+\dist{t_1}{m}).
    \end{equation}
    To prove the lemma in this case we first argue that $\angle (b_2,b_1,m)>\pi/2$.
    To this end, assume for a contradiction that $\angle (b_2,b_1,m)\leq\pi/2$.
    It is easy to verify that for a given length of~$t_1m$
    (and assuming $\angle (b_2,b_1,m)\leq\pi/2$), the fraction
    $\dist{b_1}m/(\dist{b_1}m+\dist{t_1}m)$ is
    maximized when segment $t_1m$ is perpendicular to $r_2$, and $m\in r_3$,
    and $b_1=x$. But then
    \[
    \frac{\dist{b_1}{m}}{\dist{b_1}{m}+\dist{t_1}{m}}\leq \frac{1}{1+\sin\delta/2},
    \]
    which would contradict~\eqref{ineq2}.
    Thus we indeed have $\angle (b_2,b_1,m)>\pi/2$. Hence,
    $\dist{b_2}m\geq\dist{b_1}{b_2}$, and so
    $\dist{b_1}m+\dist{b_2}m-\dist{b_1}{b_2}\geq\dist{b_1}m$.
    We can now derive
    \[
    \begin{array}{lll}
    \Phi-\Psi & = & \dist{b_1}{m}+\dist{t_1}{m}+\dist{b_2}{m}+\dist{t_2}{m} - \dist{b_1}{b_2} - \dist{t_1}{t_2} \\
              & \geq & \dist{b_1}{m}+\dist{t_1}{m}+\dist{t_2}{m} - \dist{t_1}{t_2}\quad\text{by the above}\\
              & \geq & \frac{1}{1+\sin\delta/2}\cdot \big(\dist{b_1}{m}+\dist{t_1}{m}\big)\quad\text{by~\eqref{ineq1b}~and~\eqref{ineq2}} \\
               & \geq & \big( \sin (\delta/2)\cdot(1-\cos(\delta/2)) \big) \cdot \frac{1}{1+\sin\delta/2}\cdot \big( \dist{b_1}{m}+\dist{t_1}{m}\big).   \\
    \end{array}
    \]
    Thus the lemma also holds in Case~(B).
\end{itemize}

\end{proof}

Let $\distPQ(\PP_1,\PP_2)\mydef\textrm{min}_{(p,q)\in \PP_1\times \PP_2}\dist pq$
denote the separation distance between $\PP_1$ and~$\PP_2$. Recall that $\alpha$
denotes the angle between the two common outer tangents of
$\PP_1$ and $\PP_2$; see Fig.~\ref{fig:exact}.
We are now ready to give a lower bound on the separation distance increasing in the angle $\alpha$ between the outer common tangents $\ell_3$ and $\ell_4$.
The lemma will be used when there is a positive lower bound on $\alpha$, which in turn implies a lower bound on $\distPQ(\PP_1,\PP_2)$.
\begin{lemma}\label{distanceLemma} We have
\begin{equation}\label{distIneq}
\distPQ(\PP_1,\PP_2)\ \geq\ f(\alpha) \cdot \per {\PP_1},
\end{equation}
where $f\colon [0,\pi]\longrightarrow\mathbb R$ is the increasing function
\[
    f(\varphi)
    \ \mydef\
    \frac{\sin (\varphi/4)}{1+\sin (\varphi/4)}
    \cdot
    \frac{\sin (\varphi/2)}{1+\sin (\varphi/2)}
    \cdot
\frac{1-\cos (\varphi/4)}2.
\]
\end{lemma}
\begin{proof}
The statement is trivial if $\alpha=0$ so assume $\alpha>0$.
Let $p\in \PP_1$ and $q\in \PP_2$ be points so that
$\dist pq=\distPQ(\PP_1,\PP_2)$ and assume
without loss of generality that $pq$ is a horizontal segment with $p$
being its left endpoint.
Let $\ellvert 1$ and $\ellvert 2$
be vertical lines containing $p$ and $q$, respectively.
Note that $\PP_1$ is in the closed halfplane to the left of $\ellvert 1$ and
$\PP_2$ is in the closed halfplane to the right of $\ellvert 2$.
Recall that $s_{ij}$ denotes a point on $\bd \PP_i \cap\ell_j$.
\begin{figure}
\begin{center}
\includegraphics[width=\textwidth]{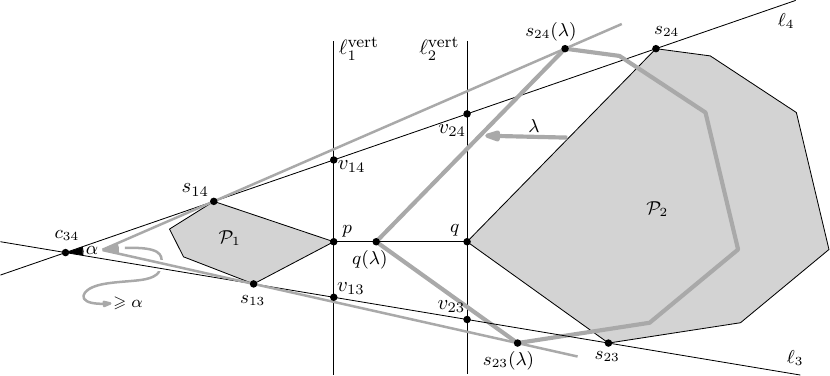}
\end{center}
\caption{Illustration for the proof of Lemma~\protect\ref{distanceLemma}.}
\label{fig:distanceLemma}
\end{figure}

\mypara{Claim:} There exist two convex polygons $\PP'_1$ and $\PP'_2$
satisfying the following conditions:
\begin{enumerate}
\item
$\PP'_1$ and $\PP'_2$ have the same outer common tangents as $\PP_1$ and $\PP_2$, namely $\ell_3$ and $\ell_4$.
\label{prop:a}

\item
$\PP'_1$ is to the left of $\ellvert 1$ and $p\in\bd\PP'_1$; and
$\PP'_2$ is to right of $\ellvert 2$ and $q\in\bd\PP'_2$.
\label{prop:b}

\item
$\per{\PP'_1}=\per{\PP_1}$.
\label{prop:c}

\item
$\per{\PP'_1}+\per{\PP'_2}\leq\per{\ch(\PP'_1\cup\PP'_2)}$.
\label{prop:d}

\item
There are points $s'_{ij}\in\PP'_i\cap\ell_j$ for all $i\in\{1,2\}$ and $j\in\{3,4\}$
such that $\bd\PP'_1(s'_{13},p)$, $\bd\PP'_1(p,s'_{14})$,
$\bd\PP'_2(s'_{24},q)$, and $\bd\PP'_2(q,s'_{23})$ each consist
of a single line segment.
\label{prop:e}

\item
Let $s'_{2j}(\lambda)\mydef s'_{2j}-(\lambda,0)$
and let $\ell'_j(\lambda)$ be the line through $s'_{1j}$ and
$s'_{2j}(\lambda)$ for $j\in\{3,4\}$.
Then $\angle(\ell'_3(\dist pq),\ell'_4(\dist pq))\geq\alpha/2$.
(Looking at Fig.~\ref{fig:distanceLemma}, one might believe that this inequality even holds for $\alpha$ instead of $\alpha/2$.
The reason for using $\alpha/2$ will be explained later.)
\label{prop:f}
\end{enumerate}

\begin{proof}[Proof of the claim]
Let $\PP_1'\mydef\PP_1$ and $\PP_2'\mydef\PP_2$, and let
$s_{ij}'$ be a point in $\PP_i'\cap\ell_j$ for all $i\in\{1,2\}$ and $j\in\{3,4\}$.
We show how to modify $\PP_1'$ and $\PP_2'$ until they
have all the required conditions. Of course,
they already satisfy conditions \ref{prop:a}--\ref{prop:d}.
We first show how to obtain condition \ref{prop:e}, namely
that $\bd\PP_1'(s_{13}',p)$
and $\bd\PP_1'(p,s_{14}')$---and similarly $\bd\PP_2'(s_{24}',q)$
and $\bd\PP_1'(q,s_{23}')$---each consist of a single line segment,
as depicted in Fig.~\ref{fig:distanceLemma}.
To this end, let $v_{ij}$ be the intersection point $\ellvert i\cap \ell_j$ for
$i\in\{1,2\}$ and $j\in\{3,4\}$. Let $s'\in s_{14}'v_{14}$ be the point such that
$\perPart{\PP_1'}{p}{s_{14}'}=\dist p{s'}+\dist{s'}{s_{14}'}$. Such a point exists since
\[
    \dist p{s_{14}'}
    \ \leq\ \perPart{\PP_1'}{p}{s_{14}'}
    \ \leq\ \dist p{v_{14}}+\dist{v_{14}}{s_{14}'}.
\]
We modify $\PP_1'$ by replacing $\bd  \PP_1'(p,s_{14}')$ with
the segments $ps'$ and $s's_{14}'$. We can now redefine
$s_{14}'\mydef s'$ so that $\bd \PP_1'(p,s_{14}')=ps_{14}'$ is a line segment.
We can modify $\PP_1'$ in a similar way to ensure that
$\bd \PP_1'(s_{13}',p)=s_{13}'p$, and we can modify $\PP_2'$ to ensure
$\bd \PP_2'(s_{24}',q)=s_{24}'q$ and $\bd \PP_2'(q,s_{23}')=qs_{23}'$.
Note that these modifications preserve conditions
\ref{prop:a}--\ref{prop:d} and that condition \ref{prop:e} is now
satisfied.

The only condition that $(\PP_1',\PP_2')$ might not satisfy is condition \ref{prop:f}.
Let $s_{2j}'(\lambda)\mydef s_{2j}'-(\lambda,0)$ and let
$\ell_j(\lambda)$ be the line through $s_{2j}'(\lambda)$ and $s_{1j}'$ for
$j\in\{3,4\}$.
Clearly, if the slopes of $\ell_3$ and $\ell_4$ have different signs (as in Fig.~\ref{fig:distanceLemma}),
the angle
$\angle(\ell_3(\lambda),\ell_4(\lambda))$ is increasing for
$\lambda\in[0,\dist pq]$, and condition \ref{prop:f} is satisfied.
However, if the slopes of $\ell_3$ and $\ell_4$ have the same sign, the angle
might decrease.

Consider the case where both slopes are positive---the other case
is analogous.
Changing $\PP_2'$ by
replacing $\bd\PP_2'(s_{23}',s_{24}')$ by the line segment
$s_{23}'s_{24}'$ makes the sum $\per{\PP_1'}+\per{\PP_2'}$ and
$\per{\ch(\PP_1'\cup\PP_2')}$ decrease equally much and hence
condition \ref{prop:d} is preserved.
This clearly has no influence on the other conditions. We thus assume that
$\PP_2'$ is the triangle $qs_{23}'s_{24}'$.
Consider what happens if we move $s_{23}'$
along the line $\ell_3$ away from $c_{34}$ with unit speed.
Then $\dist{s_{13}'}{s_{23}'}$ grows with speed exactly $1$ whereas
$\dist{q}{s_{23}'}$ grows with speed at most $1$. We therefore
preserve condition
\ref{prop:d}, and the other conditions are likewise not affected.

We now move
$s_{23}'$ sufficiently far away so that
$\angle(\ell_3,\ell_3(\dist pq))\leq\alpha/4$.
Similarly, we move $s_{24}'$ sufficiently far away from
$c_{34}$ along $\ell_4$ to ensure that
$\angle(\ell_4,\ell_4(\dist pq))\leq\alpha/4$.
It then follows that
$\angle(\ell_3(\dist pq),\ell_4(\dist pq))\geq
\angle(\ell_3,\ell_4)-\alpha/2=\alpha/2$, and condition \ref{prop:f} is
satisfied.
\end{proof}

Note that condition \ref{prop:b} in the claim implies that
$\distPQ(\PP'_1,\PP'_2)=\distPQ(\PP_1,\PP_2)=\dist pq$,
and hence inequality \eqref{distIneq} follows from
condition \ref{prop:c} if we manage to prove
$\distPQ(\PP'_1,\PP'_2)\geq f(\alpha)\cdot\per{\PP'_1}$.
Therefore,
with a slight abuse of notation, we assume from now on that
$\PP_1$ and $\PP_2$ satisfy the conditions in the claim,
where the points $s_{ij}$ play the role as
$s'_{ij}$ in conditions \ref{prop:e} and \ref{prop:f}.

We now consider a copy of $\PP_2$ that is translated horizontally to the left over a distance~$\lambda$; see Fig.~\ref{fig:distanceLemma}. Let
$s_{24}(\lambda)$, $s_{23}(\lambda)$, and $q(\lambda)$ be the translated
copies of $s_{24}$, $s_{23}$, and $q$, respectively, and let
$\ell_j(\lambda)$ be the line through $s_{1j}$ and $s_{2j}(\lambda)$
for $j\in\{3,4\}$. Furthermore, define
\[
    \Phi(\lambda)
    \ \mydef \ \dist{s_{13}}{p}+\dist{s_{14}}{p}+
    \dist{s_{23}(\lambda)}{q(\lambda)}+\dist{s_{24}(\lambda)}{q(\lambda)}
\]
and
\[
    \Psi(\lambda)
    \ \mydef \ \dist{s_{13}}{s_{23}(\lambda)}+\dist{s_{14}}{s_{24}(\lambda)}.
\]
Note that $\Phi(\lambda)=\Phi$ is constant.
By conditions \ref{prop:d} and \ref{prop:e}, we know that
\begin{equation}\label{ineq3}
\Phi\ \leq\ \Psi(0).
\end{equation}
Note that $q(\dist pq)=p$. We now apply Lemma \ref{simpleDiff} to get
\begin{equation}\label{ineq4}
\Phi-\Psi(\dist pq)
\ \geq\
\sin(\delta/2) \cdot \frac{1-\cos(\delta/2)}{1+\sin(\delta/2)}\cdot
(\dist{s_{13}}{p}+\dist{s_{14}}{p}),
\end{equation}
where $\delta\mydef\angle(\ell_3(\dist pq),\ell_4(\dist pq))$.
By condition \ref{prop:f}, we know that $\delta\geq\alpha/2$.
The function
$\varphi \longmapsto \sin(\varphi/2)\cdot \frac{1-\cos(\varphi/2)}{1+\sin(\varphi/2)}$
is increasing for $\varphi\in[0,\pi]$ and hence
inequality~\eqref{ineq4} also holds when $\delta$ is replaced by $\alpha/2$.

When $\lambda$ increases from $0$ to $\dist pq$
with unit speed,
the value $\Psi(\lambda)$ decreases with speed at most~$2$, i.e.,
$\Psi(\lambda)\geq\Psi(0)-2\lambda$.
Using this and inequalities \eqref{ineq3} and \eqref{ineq4}, we get
    \[
    2\dist pq \ \geq \ \Psi(0)-\Psi(\dist pq)\ \geq\ \Phi-\Phi+
    \sin(\alpha/4) \cdot \frac{1-\cos(\alpha/4)}{1+\sin(\alpha/4)}\cdot
(\dist{s_{13}}{p}+\dist{s_{14}}{p}),
\]
and we conclude that
\begin{equation}\label{ineq5}
\dist pq\ \geq\
\frac 12\cdot\sin(\alpha/4)\cdot\frac{1-\cos(\alpha/4)}{1+\sin(\alpha/4)}\cdot
(\dist{s_{13}}{p}+\dist{s_{14}}{p}).
\end{equation}

By the triangle inequality,
$\dist{s_{13}}{p}+\dist{s_{14}}{p}\geq \dist{s_{13}}{s_{14}}$.
Furthermore, for a given length of $s_{13}s_{14}$,
the fraction
$\dist{s_{13}}{s_{14}} / (\dist{s_{14}}{c_{34}}+\dist{c_{34}}{s_{13}})$
is minimized when $s_{13}s_{14}$ is perpendicular to the angular bisector of
$\ell_3$ and $\ell_4$. (Recall that $c_{34}$ is the intersection point of the
outer common tangents $\ell_3$ and $\ell_4$; see Fig.~\ref{fig:distanceLemma}.)
Hence
\begin{equation}\label{ineqs13}
    \dist{s_{13}}{s_{14}}
    \ \geq\
    \sin(\alpha/2)\cdot\left(\dist{s_{14}}{c_{34}}+\dist{c_{34}}{s_{13}}\right).
\end{equation}
We now conclude
\[
\begin{array}{lll}
\dist{s_{13}}{p}+\dist{s_{14}}{p}
    &  = & \frac{\sin(\alpha/2)}{1+\sin(\alpha/2)} \cdot \Big( \frac{\dist{s_{13}}{p}+\dist{s_{14}}{p}}{\sin(\alpha/2)} + \dist{s_{13}}{p}+\dist{s_{14}}{p} \Big) \\[2mm]
    &  \geq & \frac{\sin(\alpha/2)}{1+\sin(\alpha/2)}
              \cdot \Big( \frac{\dist{s_{13}}{s_{14}}}{\sin(\alpha/2)} + \dist{s_{13}}{p}+\dist{s_{14}}{p} \Big)
              \quad\text{triangle inequality} \\[2mm]
    & \geq & \frac{\sin(\alpha/2)}{1+\sin (\alpha/2)}
             \cdot \Big( \dist{s_{14}}{c_{34}}+\dist{c_{34}}{s_{13}}+\dist{s_{13}}{p}+\dist{s_{14}}{p} \Big)
             \quad\text{by \eqref{ineqs13}} \\[2mm]
    & \geq & \frac{\sin(\alpha/2)}{1+\sin(\alpha/2)}\cdot\per{\PP_1},
\end{array}
\]
where the last inequality follows because $\PP_1$ is fully contained in the
quadrilateral~$s_{14},c_{34},x_{13},p$.
The statement~\eqref{distIneq} in the lemma now follows from \eqref{ineq5}.
\end{proof}
We are now ready to prove Theorem~\ref{th:separation-property}.
\begin{proof}[Proof of Theorem~\ref{th:separation-property}]
If the separation angle of $P_1$ and $P_2$ is at least $\pi/6$, we are done.
Otherwise,
Lemma \ref{bigAngles} gives that $\alpha>\pi/2$,
and Lemma \ref{distanceLemma} gives that
$\distPQ(\PP_1,\PP_2)\ \geq\ f(\pi/2) \cdot \per {\PP_1} \geq
(1/250)\cdot\min(\per{\PP_1},\per{\PP_2})$.
\end{proof}

\subsection{The algorithm}\label{subse:alg}
Theorem~\ref{th:separation-property} suggests to distinguish two cases when computing
an optimal partition: the case when the separation angle is
large (namely at least $\pi/6$) and the case when the separation distance is
large (namely at least $\csep\cdot \min (\myper(P_1),\myper(P_2))$).
As we will see, the first case can be handled in
$O(n \log n)$ time and the second case in
$O(n\log^2 n)$ time, leading to the following theorem.
\begin{theorem}\label{th:main}
Let $P$ be a set of $n$ points in the plane. Then we can compute a partition $(P_1,P_2)$
of $P$ that minimizes $\myper(P_1)+\myper(P_2)$ in $O(n\log^2 n)$ time
using $O(n\log^2 n)$ space.
\end{theorem}

\subsubsection{The best partition with large separation angle}\label{sec:sep}
Define the \emph{orientation} of a line~$\ell$, denoted by $\phi(\ell)$, to be the
counterclockwise angle that~$\ell$ makes with the positive $y$-axis.
If the separation angle of $P_1$ and $P_2$ is at least $\pi/6$, then there must
be a line $\ell$ separating $P_1$ from $P_2$ that does not contain any point
from $P$ and such that $\phi(\ell)= j\cdot \pi/7$ for some~$j\in\{0,1,\ldots,6\}$.
For each of these seven orientations we can compute the best partition
in $O(n\log n)$ time, as explained next.

Without loss of generality, consider separating lines~$\ell$ with $\phi(\ell)=0$,
that is, vertical separating lines. Let $X$ be the set of all $x$-coordinates of the
points in~$P$. For any $x$-value $x\in X$ define $P_1(x) \mydef \{ p\in P \mid p_x \leq x \}$, where
$p_x$ denotes the $x$-coordinate of a point~$p$, and define $P_2(x) \mydef P\setminus P_1(x)$.
Our task is to find the best partition of the form~$(P_1(x),P_2(x))$ over all $x\in X$.
To this end we first compute the values $\myper(P_1(x))$ for all $x\in X$
in $O(n \log n)$ time in total, as follows. We compute the lengths of the upper hulls
of the point sets~$P_1(x)$, for all $x\in X$, using Graham's scan~\cite{bcko-cgaa-08},
and we compute the lengths of the lower hulls in a second scan.
(Graham's scan goes over the points from left to right and maintains
the upper (or lower) hull of the encountered points; it is trivial to extend
the algorithm so that it also maintains the length of the hull.)
By combining the lengths of the upper and lower hulls, we get the values $\myper(P_1(x))$.

Computing the values $\myper(P_2(x))$ can be done similarly, after which we can easily
find the best partition of the form~$(P_1(x),P_2(x))$ in $O(n)$ time.
Thus the best partition with large separation angle can be found in $O(n\log n)$ time.

\subsubsection{The best partition with large separation distance}\label{sec:dist}
Next we show how to compute the best partition with large separation distance.
We assume without loss of generality that $\myper(P_2)\leq \myper(P_1)$.
It will be convenient to treat the case where $P_2$ is a singleton separately.
\begin{lemma}\label{le:singleton}
The point $p\in P$ minimizing $\myper(P\setminus \{p\})$
can be computed using $O(n\log n)$ time.
\end{lemma}
\begin{proof}
The point~$p$ we are looking for must be a vertex of~$\ch(P)$.
First we compute $\ch(P)$ in $O(n\log n)$ time~\cite{bcko-cgaa-08}.
Let $v_0,v_1,\ldots,v_{m-1}$ denote the vertices of~$\ch(P)$ in counterclockwise order.
Let $\Delta_i$ be the triangle with vertices $v_{i-1} v_i v_{i+1}$
(with indices taken modulo~$m$) and let
$P_i$ denote the set of points lying inside~$\Delta_i$,
excluding $v_i$ but including $v_{i-1}$ and $v_{i+1}$.
Note that any point $p\in P$ is present in
at most two sets~$P_i$. Hence, $\sum_{i=0}^m |P_i|= O(n)$.
It is not hard to compute the sets $P_i$ in $O(n\log n)$ time in total.
After doing so, we compute all convex hulls~$\ch(P_i)$
in $O(n\log n)$ time in total. Since
\[
\myper(P\setminus\{v_i\})
  = \myper(P) - |v_{i-1}v_i| - |v_i v_{i+1}| +  \myper(P_i) - |v_{i-1}v_{i+1}|,
\]
we can now find the point $p$ minimizing $\myper(P\setminus \{p\})$
in $O(n)$ time.
\end{proof}

It remains to compute the best partition $(P_1,P_2)$ with
$\myper(P_2)\leq \myper(P_1)$ whose separation distance is
at least~$\csep \cdot\myper(P_2)$ and where $P_2$ is not a singleton.
Let $(P_1^*,P_2^*)$ denote this partition. Define the \emph{size}
of a square\footnote{Whenever we speak of squares, we always mean axis-parallel squares.}~$\sigma$ to be its edge length.
A square $\sigma$
is a \emph{good square} if
(i) $P_2^* \subset \sigma$, and
(ii) $\size(\sigma)\leq c^* \cdot \myper(P^*_2)$, where $c^* \mydef 18$.
Our algorithm globally works as follows.
\begin{enumerate}
\item Compute a set $S$ of $O(n)$ squares such that $S$ contains a good square.
\item For each square $\sigma\in S$, construct a set $H_{\sigma}$ of $O(1)$
      halfplanes such that the following holds: if $\sigma\in S$ is a good square
      then there is a halfplane $h\in H_\sigma$ such that $P^*_2 = P(\sigma\cap h)$,
      where $P(\sigma\cap h) \mydef P\cap (\sigma\cap h)$.
\item For each pair $(\sigma,h)$ with $\sigma\in S$ and $h\in H_\sigma$, compute
      $\myper(P\setminus P(\sigma\cap h)) + \myper(P(\sigma\cap h))$, and
      report the partition $(P\setminus P(\sigma\cap h), P(\sigma\cap h))$
      that gives the smallest sum.
\end{enumerate}
\mypara{Step 1: Finding a good square.}
To find a set $S$ that contains a good square, we first construct a set
$\Sbase$ of so-called \emph{base squares}. The set $S$ will then be obtained
by expanding the base squares appropriately.

We define a base square~$\sigma$ to be \emph{good} if
(i) $\sigma$ contains at least one point from $P^*_2$, and
(ii) $c_1 \cdot\diam(P^*_2) \leq \size(\sigma) \leq c_2 \cdot\diam(P^*_2)$,
where $c_1 \mydef 1/4$ and $c_2\mydef 4$
and $\diam(P^*_2)$ denotes the diameter of~$P^*_2$.
Note that $2\cdot\diam(P^*_2)\leq\per{P^*_2}\leq 4\cdot\diam(P^*_2)$.
For a square~$\sigma$, define $\expand{\sigma}$ to be the square
with the same center as $\sigma$ and whose size is
$\left(1+2/c_1\right)\cdot\size(\sigma)$.
\begin{lemma}\label{le:base-square}
If $\sigma$ is a good base square then $\expand{\sigma}$ is a good square.
\end{lemma}
\begin{proof}
The distance from any point in $\sigma$ to the boundary of $\expand{\sigma}$
is at least
\[
\frac{\size(\expand{\sigma})-\size(\sigma)}{2}
    \ \geq\ \diam(P^*_2).
\]
Since $\sigma$ contains a point from $P^*_2$,
it follows that
$P^*_2\subset\expand{\sigma}$.
Since $\size(\sigma) \leq c_2 \cdot \diam(P^*_2)$, we have
\[
\size(\expand{\sigma})
    \ \leq \ (2/c_1+1)\cdot c_2 \cdot\diam(P^*_2)
    \ = \ 36 \cdot \diam(P^*_2)
    \ \leq\ c^*\cdot \myper(P^*_2).
\]
\end{proof}

To obtain $S$ it thus suffices to construct a set $\Sbase$ that contains a good base square.
To this end we first build a compressed quadtree for $P$. For completeness we briefly
review the definition of compressed quadtrees; see also Fig.~\ref{fi:base-squares-types}~(left).

Assume without loss of generality that $P$ lies in the interior of the unit square $U\mydef [0,1]^2$.
Define a \emph{canonical square} to be any square that
can be obtained by subdividing~$U$ recursively into quadrants.
A \emph{compressed quadtree}~\cite{h-11} for $P$ is a hierarchical subdivision of $U$,
defined as follows.
In a generic step of the recursive process we are given a
canonical square $\sigma$ and the set $P(\sigma) \mydef P\cap \sigma$ of
points inside~$\sigma$.
(Initially $\sigma =U$ and $P(\sigma)=P$.) \\
\begin{itemize}
\item If $|P(\sigma)|\leq 1$ then the recursive process stops and $\sigma$
      is a square in the final subdivision.
\item Otherwise there are two cases.
      Consider the four quadrants of $\sigma$.
      The first case is that at least two of these quadrants contain points from $P(\sigma)$.
      (We consider the quadrants to be closed on the left and bottom side,
       and open on the right and top side, so a point is contained in a unique quadrant.)
      In this case we partition $\sigma$ into its four quadrants---we call this
      a \emph{quadtree split}---and recurse on each quadrant.
      The second case is that all points from $P(\sigma)$ lie inside the same quadrant.
      In this case we compute the smallest canonical square, $\sigma'$, that
      contains $P(\sigma)$ and we partition $\sigma$ into two regions: the
      square $\sigma'$ and the so-called \emph{donut region}~$\sigma\setminus \sigma'$.
      We call this a \emph{shrinking step}. After a shrinking step we only
      recurse on the square $\sigma'$, not on the donut region. \\
\end{itemize}
A compressed quadtree for a set of $n$ points can be computed in $O(n\log n)$ time
in the appropriate model of computation\footnote{In particular we need to be able to compute
the smallest canonical square containing two given points in~$O(1)$ time. See the book by
Har-Peled~\cite{h-11} for a discussion.}\label{fn:quadtree}~\cite{h-11}.
The idea is now as follows. Let $p,p'\in P^*_2$ be a pair of points defining~$\diam(P_2^*)$.
The compressed quadtree hopefully allows us to zoom in until we have a square in the
compressed quadtree that contains $p$ or $p'$ and whose size is roughly equal to~$|pp'|$.
Such a square will be then a good base square.
Unfortunately this does not always work since $p$ and $p'$
can be separated too early. We therefore have to proceed more carefully:
we need to add five types of base squares to $\Sbase$, as explained next and illustrated in
Fig.~\ref{fi:base-squares-types}~(right).
\begin{figure}
\begin{center}
\includegraphics[width=1.0\textwidth]{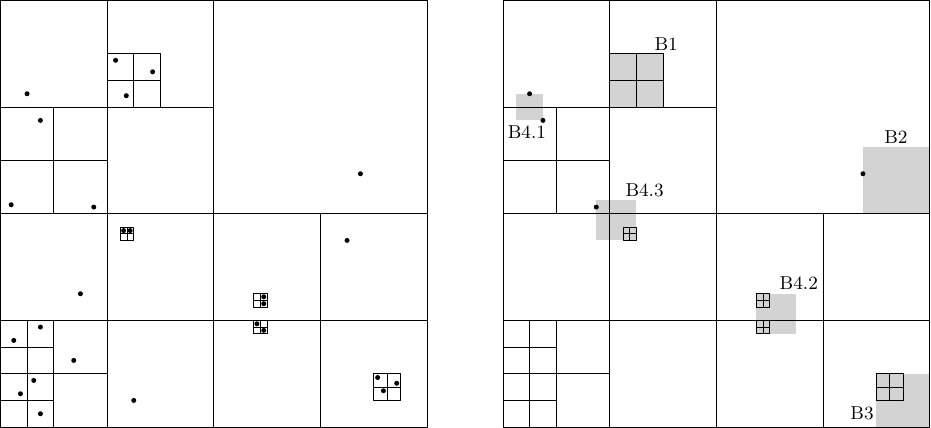}
\end{center}
\caption{A compressed quadtree and some of the base squares generated from it.
In the right figure, only the points are shown that are relevant
for the shown base squares.}
\label{fi:base-squares-types}
\end{figure}

\begin{description}
\item[(B1)] Any square $\sigma$ that is generated during the recursive construction---note
      that this not only refers to squares in the final subdivision---is put into~$\Sbase$.
\item[(B2)] For each point $p\in P$ we add a square $\sigma_p$ to $\Sbase$, as follows.
      Let $\sigma$ be the square of the final subdivision that contains~$p$. Then
      $\sigma_p$ is a smallest square
      that contains $p$ and that shares a corner with~$\sigma$.
\item[(B3)] For each square $\sigma$ that results from a shrinking step we add an extra
      square~$\sigma'$ to $\Sbase$, where $\sigma'$ is the smallest square that contains~$\sigma$ and that shares a corner with the parent square of $\sigma$.
\item[(B4)] For any two regions in the final subdivision that touch each other---we also
      consider two regions to touch if they only share a vertex---we
      add at most one square to $\Sbase$, as follows.
      If one of the regions is an empty square, we do not add anything
      for this pair. Otherwise we have three cases.
      \begin{description}
      \item[(B4.1)] If both regions are non-empty squares containing single
            points~$p$ and $p'$, respectively, then we add a smallest
            enclosing square for the pair of points $p,p'$ to $\Sbase$.
      \item[(B4.2)] If both regions are donut regions, say $\sigma_1\setminus\sigma_1'$
            and $\sigma_2\setminus\sigma_2'$, then we add a smallest enclosing square
            for the pair $\sigma'_1,\sigma'_2$ to $\Sbase$.
      \item[(B4.3)] If one region is a non-empty square containing a single point $p$
            and the other is a donut region $\sigma\setminus\sigma'$,
            then we add a smallest enclosing square
            for the pair $p,\sigma'$ to $\Sbase$.
      \end{description}
\end{description}
\begin{lemma}
The set $\Sbase$ has size $O(n)$ and contains a good base square. Furthermore, $\Sbase$
can be computed in $O(n\log n)$ time.
\end{lemma}
\begin{proof}
A compressed quadtree has size $O(n)$ so we have $O(n)$ base squares of type~(B1) and~(B3).
Obviously there are $O(n)$ base squares of type~(B2). Finally, the number
of pairs of final regions that touch is $O(n)$---this follows because we have a planar
rectilinear subdivision of total complexity~$O(n)$---and so the number of base squares
of type~(B4) is $O(n)$ as well. The fact that we can compute $\Sbase$ in $O(n\log n)$ time
follows directly from the fact that we can compute the compressed quadtree in $O(n\log n)$
time~\cite{h-11}.
\medskip

It remains to prove that $\Sbase$ contains a good base square. We call a square~$\sigma$
\emph{too small} when $\size(\sigma) < c_1 \cdot \diam(P^*_2)$  and \emph{too large} when
$\size(\sigma) > c_2\cdot\diam(P^*_2)$; otherwise we say that $\sigma$ has the \emph{correct size}.
Let $p,p'\in P^*_2$ be two points with $|pp'|=\diam(P^*_2)$, and consider a smallest square
$\sigma_{p,p'}$, in the compressed quadtree that contains both~$p$ and~$p'$.
Note that $\sigma_{p,p'}$ cannot be too small, since $c_1=1/4<1/\sqrt{2}$.
If $\sigma_{p,p'}$ has the correct size, then we are done since it is
a good base square of type~(B1). So now suppose $\sigma_{p,p'}$ is too large.

Let $\sigma_0,\sigma_1,\ldots,\sigma_k$ be the sequence of squares
in the recursive subdivision of $\sigma_{p,p'}$ that contain~$p$;
thus $\sigma_0=\sigma_{p,p'}$ and $\sigma_k$ is a square in the final subdivision.
Define $\sigma'_0,\sigma'_1,\ldots,\sigma'_{k'}$
similarly, but now for~$p'$ instead of~$p$. Suppose that none of these
squares has the correct size---otherwise we have a good base square of type~(B1).
There are three cases.
\begin{itemize}
\item \emph{Case~(i): $\sigma_k$ and $\sigma'_{k'}$ are too large.}
      \\[1mm]
      We claim that $\sigma_k$ touches $\sigma'_{k'}$. To see this, assume without loss
      of generality that $\size(\sigma_k)\leq\size(\sigma'_{k'})$. If
      $\sigma_k$ does not touch $\sigma'_{k'}$ then $|pp'| \geq \size(\sigma_k)$,
      which contradicts the assumption that $\sigma_k$ is too large.
      Hence, $\sigma_k$ indeed touches $\sigma'_{k'}$. But then we have a base square
      of type~(B4.1) for the pair $p,p'$ and since $|pp'|=\diam(P^*_2)$ this
      is a good base square.
\item \emph{Case~(ii): $\sigma_k$ and $\sigma'_{k'}$ are too small.}
      \\[1mm]
      In this case there are indices $0<j\leq k$ and $0<j'\leq k'$ such that
      $\sigma_{j-1}$ and $\sigma'_{j'-1}$ are too large and $\sigma_{j}$
      and $\sigma'_{j'}$ are too small.
      Note that this implies that both $\sigma_j$ and $\sigma'_{j'}$ result
      from a shrinking step, because $c_1 < c_2/2$ and so the quadrants
      of a too-large square cannot be too small.
      We claim that $\sigma_{j-1}$ touches $\sigma'_{j'-1}$. Indeed, similarly to Case~(i),
      if $\sigma_{j-1}$ and $\sigma'_{j'-1}$ do not touch
      then $|pp'| > \min(\size(\sigma_{j-1}),\size(\sigma'_{j'-1}))$,
      contradicting the assumption that both $\sigma_{j-1}$ and $\sigma'_{j'-1}$ are too large.
      We now have two subcases.
      \begin{itemize}
      \item The first subcase is that the donut region~$\sigma_{j-1}\setminus\sigma_{j}$
            touches the donut region~$\sigma'_{j'-1}\setminus\sigma_{j'}$. Thus a
            smallest enclosing square for $\sigma_j$ and $\sigma'_{j'}$ has been put
            into $\Sbase$ as a base square of type~(B4.2). Let $\sigma^*$ denote
            this square. Since the segment~$pp'$ is contained in $\sigma^*$ we have
            \[
            c_1\cdot \diam(P^*_2) \ < \ \diam(P^*_2)/ \sqrt{2} \ = \ |pp'| / \sqrt{2} \ \leq \  \size(\sigma^*).
            \]
            Furthermore, since $\sigma_j$ and $\sigma'_{j'}$ are too small
            we have
            \begin{align} \label{eq:first-subcase}
            \size(\sigma^*)
            \ & \leq \ \size(\sigma_j) + \size(\sigma'_{j'} ) + |pp'|
            \ \leq \ 3\cdot \diam(P^*_2)
            \\ & < \ c_2\cdot \diam(P^*_2), \nonumber
            \end{align}
            and so $\sigma^*$ is a good base square. 
      \item The second subcase is that $\sigma_{j-1}\setminus\sigma_{j}$
            does not touch~$\sigma'_{j'-1}\setminus\sigma_{j'}$. This can only
            happen if $\sigma_{j-1}$ and $\sigma'_{j'-1}$ just share a single corner, $v$.
            Observe that $\sigma_j$ must lie in the quadrant of $\sigma_{j-1}$ that
            has $v$ as a corner, otherwise $|pp'|\geq \size(\sigma_{j-1})/2$
            and $\sigma_{j-1}$ would not be too large.
            Similarly,  $\sigma'_{j'}$ must lie in the quadrant of $\sigma'_{j'-1}$ that
            has $v$ as a corner. Thus the base squares of type~(B3) for
            $\sigma_j$ and $\sigma'_{j'}$ both have $v$ as a corner. Take the largest
            of these two base squares, say~$\sigma_j$. For this square~$\sigma^*$ we have
            \[
            c_1\cdot \diam(P^*_2) \ < \ \diam(P^*_2)/2\sqrt{2} \ = \ |pp'| / 2\sqrt{2} \ \leq \ \size(\sigma^*),
            \]
            since $|pp'|$ is contained in a square of twice the size of $\sigma^*$.
            Furthermore, since $\sigma_j$ is too small and $|pv|<|pp'|$ we have
            \begin{equation} \label{eq:second-subcase}
            \size(\sigma^*)
            \ \leq \ \size(\sigma_j) + |pv|
            \ \leq \ (c_1+1)\cdot \diam(P^*_2)
            \ < \ c_2\cdot \diam(P^*_2).
            \end{equation}
            Hence, $\sigma^*$ is a good base square.
      \end{itemize}
\item \emph{Case~(iii): neither (i) nor (ii) applies.} \\[1mm]
      In this case $\sigma_k$ is too small and $\sigma'_{k'}$ is too large (or vice versa).
      Thus there must be an index $0<j\leq k$ such that
      $\sigma_{j-1}$ is too large and $\sigma_{j}$ is too small.
      We can now follow a similar reasoning as in Case~(ii):
      First we argue that $\sigma_j$ must have resulted from
      a shrinking step and that $\sigma_{j-1}$ touches~$\sigma'_{k'}$.
      Then we distinguish two subcases, namely where the
      donut region~$\sigma_j\setminus \sigma_{j-1}$ touches~$\sigma'_{k'}$
      and where it does not touch~$\sigma'_{k'}$.
      The arguments for the two subcases are similar to the subcases in Case~(ii),
      with the following modifications.
      In the first subcase we use base squares of type~(B4.3)
      and in \eqref{eq:first-subcase} the term $\size(\sigma'_{j'})$ disappears;
      in the second subcase we use a type~(B3) base square
      for $\sigma_{j}$ and a type~(B2) base square for $p'$,
      and when the base square for $p'$ is larger than the
      base square for $\sigma_{j}$ then \eqref{eq:second-subcase} becomes
      $\size(\sigma^*) \leq 2\;|p'v| < c_2\cdot \diam(P^*_2)$.
\end{itemize}

\end{proof}

\mypara{Step 2: Generating halfplanes.}
Consider a good square $\sigma \in S$. Let $Q_\sigma$ be a set of $4\cdot c^*/\csep+1=18001$ points placed equidistantly around the boundary of $\sigma$.
Note that the distance between two neighbouring points in $Q_\sigma$ is
less than $\csep/c^* \cdot \size(\sigma)$.
For each pair $q_1,q_2$
of points in $Q_\sigma$, add to $H_\sigma$ the two halfplanes defined
by the line through $q_1$ and $q_2$. 

\begin{lemma}\label{le:hyperplanes} For any good square $\sigma \in S$, there is a halfplane $h\in H_\sigma$ such that $P^*_2 = P(\sigma\cap h)$.
\end{lemma}

First a remark: We do not claim that the line $\ell$ bounding the halfplane $h$ separates $P^*_1$ and $P^*_2$ globally, but only in $\sigma$---indeed, $\ell$ might intersect $\ch(P^*_1)$.

\begin{proof}
In the case where $\sigma\cap P^*_1=\emptyset$, two points in $Q_\sigma$ from the same
edge of $\sigma$ define a halfplane $h$ such that $P^*_2 = P(\sigma\cap h)$,
so assume that $\sigma$ contains one or more points from
$P^*_1$.

We know that the separation distance between $P^*_1$ and $P^*_2$ is at least $\csep \cdot \myper(P^*_2)$. 
Moreover, $\size(\sigma)\leq c^* \cdot \myper(P^*_2)$. 
Hence,
there is an empty open strip $O$ with a width of at least $\csep/c^* \cdot \size(\sigma)$
separating $P^*_2$ from $P^*_1$.
Since $\sigma$ contains a point from $P^*_1$, we know that
$\sigma\setminus O$ consists of two pieces and that
the part of the
boundary of $\sigma$ inside~$O$ consists of two disjoint portions $B_1$ and $B_2$
each of length at least $\csep/c^* \cdot \size(\sigma)$. Hence the sets
$B_1\cap Q_\sigma$ and $B_2\cap Q_\sigma$ contain points $q_1$ and $q_2$,
respectively, that define a halfplane $h$ as desired.
\end{proof}

\mypara{Step 3: Evaluating candidate solutions.}
In this step we need to compute for each pair $(\sigma,h)$ with $\sigma\in S$
and $h\in H_\sigma$, the value $\myper(P\setminus P(\sigma\cap h)) + \myper(P(\sigma\cap h))$.
Given a set $\mathcal O$ of $k$ orientations, Oh and Ahn~\cite{oa-18} described how to create a data structure using $O(nk^3\log^2 n)$ time and space to answer queries of the following type in time $O(k\log^2 n)$:
Given a convex polygon $Q$ where each edge has an orientation in~$\mathcal O$, what is $\myper(P\cap Q)$?
In our case, we need to compute the perimeter of the points in canonical $5$-gons and their complements, i.e., $\myper(P(\sigma\cap h))$ and $\myper(P\setminus P(\sigma\cap h))$ for a given pair $(\sigma,h)$.
Recall that the bounding lines of the halfplanes $h$ we must process have $O(1)$ different orientations.
For each such orientation $o$, we make an instance of the data structure of Oh and Ahn which has as orientations $\mathcal O$ the two axis-parallel directions and $o$.
We can then clearly compute $\myper(P(\sigma\cap h))$ in time~$O(\log^2 n)$.
Note that the complement $P\setminus P(\sigma\cap h)$ is the disjoint union of the points in four axis-parallel rectangles and the complementary canonical $5$-gon $\sigma\setminus(\sigma\cap h)$.
For each of these four rectangles and the $5$-gon, we can compute the convex hull of the points inside it in $O(\log^2 n)$ time, using the data structure of Oh and Ahn~\cite{oa-18}.
This gives us five convex hulls, represented as balanced trees.
We can then compute, for each pair of convex hulls, the outer common tangents in $O(\log n)$ time~\cite[Lemma 3]{oa-18}, from which we can compute the overall convex hull and its perimeter.
The total time to compute $\myper( P \setminus (P(\sigma\cap h))$ is thus likewise $O(\log^2 n)$.

We thus obtain the following result, which finishes the proof of Theorem~\ref{th:main}.
\begin{lemma}\label{le:step3prime}
Step 3 can be performed in $O(n\log^2 n)$ time and space.
\end{lemma}

\section{The approximation algorithm}\label{se:approx-alg}

\begin{theorem}\label{thm:approx}
Let $P$ be a set of $n$ points in the plane and let $(P_1^*,P_2^*)$ be a partition of~$P$
minimizing $\per{P_1^*}+\per{P_2^*}$. Suppose we have an exact algorithm
for the minimum perimeter-sum problem running in $T(k)$ time for instances with $k$ points. Then for any given $\eps>0$
we can compute a partition $(P_1,P_2)$ of $P$ such that
$\per{P_1}+\per{P_2} \leq (1+\eps) \cdot  \big(\per{P_1^*}+\per{P_2^*}\big)$ in $O(n+T(1/\eps^2))$ time.
\end{theorem}

\begin{proof}
Consider the axis-parallel bounding box $B$ of $P$. Let $w$ be the width of~$B$ and let
$h$ be its height. Assume without loss of generality that $w\geq h$.
Our algorithm works in two steps.
\begin{itemize}
\item \emph{Step 1: Check if $\per{P_1^*}+\per{P_2^*}\leq w/16$. If so, compute the exact solution.} \\[2mm]
    We partition $B$ vertically into four strips with width~$w/4$, denoted $B_1$, $B_2$, $B_3$,
    and $B_4$ from left to right. If $B_2$ or $B_3$ contains a point from $P$,
    we have $\per{P_1^*}+\per{P_2^*}\geq w/2>w/16$ and we go to Step~2.
    If $B_2$ and $B_3$ are both empty, we consider two cases.
    \begin{itemize}
    \item \emph{Case~(i): $h\leq w/8$.} \\[2mm]
         In this case we simply return the partition $(P\cap B_1,P\cap B_4)$.
        To see that this is optimal, we first note that any subset $P'\subset P$
        that contains a point from $B_1$ as well as a point from $B_4$
        has $\per{P'}\geq 2\cdot (3w/4)=3w/2$. On the other hand,
        $\per{P\cap B_1}+\per{P\cap B_4} \leq 2\cdot (w/2+ 2h) \leq 3w/2$.
    \item \emph{Case~(ii): $h> w/8$.} \\[2mm]
        We partition $B$ horizontally into four rows with height~$h/4$,
        numbered $R_1$, $R_2$, $R_3$, and $R_4$ from bottom to top.
        If $R_2$ or $R_3$ contains a point from $P$, we have
        $\per{P_1^*}+\per{P_2^*}\geq h/2>w/16$, and we go to Step~2.
        If $R_2$ and $R_3$ are both empty, we overlay the vertical and
        the horizontal partitioning of $B$ to get a
        $4\times 4$ grid of cells $C_{ij}\mydef B_i\cap R_j$ for $i,j\in\{1,\ldots,4\}$.
        We know that only the corner cells $C_{11},C_{14},C_{41},C_{44}$
        contain points from $P$. If three or four corner cells are non-empty,
        $\per{P_1^*}+\per{P_2^*}\geq 6h/4>w/16$, and we go to Step~2. Hence, we may without loss of
        generality assume that any point of $P$ is in $C_{11}$ or $C_{44}$.
        We now return the partition $(P\cap C_{11},P\cap C_{44})$, which
        is easily seen to be optimal.
    \end{itemize}
\item \emph{Step 2: Handle the case where $\per{P_1^*}+\per{P_2^*}> w/16$.} \\[2mm]
    The idea is to compute a subset~$\widehat{P}\subset P$ of size $O(1/\eps^2)$
    such that an exact solution to the minimum perimeter-sum problem on~$\widehat{P}$
    can be used to obtain a $(1+\eps)$-approximation for the problem on~$P$.

    We subdivide $B$ into $O(1/\eps^2)$ rectangular cells of width and height
    at most $c\mydef \eps w / (64\pi\sqrt 2)$.
    For each cell $C$ where $P\cap C$ is non-empty we pick an arbitrary point
    in $P\cap C$, and we let $\widehat P$ be the set of selected points.
    For a point $p\in \widehat P$, let $C(p)$ be the cell containing $p$.
    Intuitively, each point $p\in \widehat P$ represents all the points $P\cap C(p)$.
    Let $(\widehat P_1,\widehat P_2)$ be a partition of
    $\widehat P$ that minimizes $\per{\widehat P_1}+\per{\widehat P_2}$.
    We assume we have an algorithm that can compute such an optimal partition
    in $T(|\widehat{P}|)$ time. For $i=1,2$, define
    \[
    P_i\mydef \bigcup_{p\in \widehat P_i} P\cap C(p).
    \]
    Our approximation algorithm returns the partition $(P_1,P_2)$.
    (Note that the convex hulls of $P_1$ and $P_2$ are not necessarily disjoint.)
    It remains to prove the approximation ratio.
\begin{figure}
\begin{center}
\includegraphics{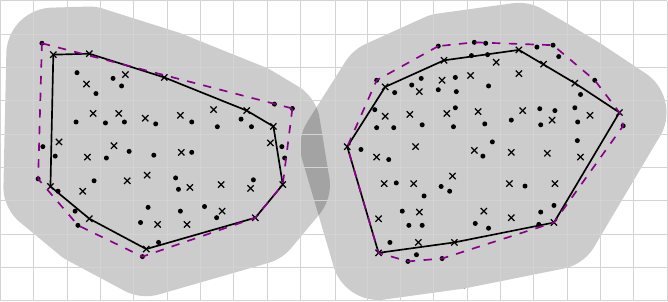}
\end{center}
\caption{The crossed points are the points of $\widehat P$.
The left gray region is $\widetilde P_1$ and the right
gray region is $\widetilde P_2$. The left dashed polygon is the convex hull of $P_1$
and the right dashed polygon is the convex hull of $P_2$.}
\label{fi:coreset}
\end{figure}

    First, note that $\per{\widehat P_1}+\per{\widehat P_2}\leq\per{P_1^*}+\per{P_2^*}$ since
    $\widehat P\subseteq P$.
    For $i=1,2$, let $\widetilde P_i$ consist of all points in the plane (not only points
    in $P$) within a distance of at most $c\sqrt{2}$ from $\ch(\widehat P_i)$.
    In other words, $\widetilde P_i$ is the Minkowksi sum of $\ch(\widehat P_i)$
    with a disk $D$ of radius $c\sqrt{2}$ centered at the origin;
    see Fig.~\ref{fi:coreset}.
    Note that if
    $p\in \widehat P_i$, then $q\in\widetilde P_i$ for any $q\in P\cap C(p)$,
    since any two points in $C(p)$ are at most $c\sqrt{2}$ apart from each other.
    Therefore $P_i\subset \widetilde P_i$ and hence $\per{P_i}\leq \per{\widetilde P_i}$.
    Note also that
    $\per{\widetilde P_i}=\per{\widehat P_i}+2 c\pi\sqrt 2$.
    These observations yield
    \[
    \begin{array}{lll}
    \per{P_1}+\per{P_2}
        & \leq  & \per{\widetilde P_1}+\per{\widetilde P_2} \\
        &   =   & \per{\widehat P_1}+\per{\widehat P_2} + 4 c\pi\sqrt 2 \\
        & \leq  & \per{P_1^*}+\per{P_2^*} + 4 c\pi\sqrt 2 \\
        &  =    & \per{P_1^*}+\per{P_2^*} + 4 \pi\sqrt 2 \cdot \left( \eps w / (64\pi\sqrt 2) \right) \\
        & \leq  & \per{P_1^*}+\per{P_2^*} + \eps w/16 \\
        & \leq  & (1+\eps) \cdot (\per{P_1^*}+\per{P_2^*}).
    \end{array}
    \]
   As all the steps can be done in linear time,
   the time complexity of the algorithm is $O(n+T(n_\eps))$
   for some $n_\eps=O(1/\eps^2)$.
\end{itemize}
\end{proof}

\section{Concluding remarks}

We note that in the exact algorithm, for each of the $O(n)$ base squares $\sigma\in S$, the number of values $\myper(P(\sigma\cap h))$ that we query is approximately $2\cdot {4\choose 2}\cdot 4500^2\approx 2.4\cdot 10^8$.
Although it is surely possible to modify the algorithm to get smaller constants (to which we made no attempt), we expect that the algorithm will remain impractical.

Consider the degenerate case where all the input points $P$ are on the $x$-axis.
Then the minimum-perimeter sum problem reduces to the well-known maximum-gap problem, where the goal is to find the largest difference between two consecutive numbers in sorted order.
Lee and Wu~\cite{lw-86} gave a lower bound of $\Omega(n\log n)$ for that problem in the algebraic computation tree model, which therefore also holds for the minimum-perimeter sum problem in that model.

The question by Mitchell and Wynters~\cite{mw-91} about the existence of sub-quadratic algorithms for the minimum-perimeter maximum, minimum-area sum, and minimum-area maximum problems remain interesting open problems.
To our knowledge, the only published algorithm for any of these problems is the $O(n^4\log n)$-time algorithm by Bae \etal~\cite{bcess-2016} for the minimum-area sum problem, since the algorithms by Mitchell and Wynters consider line partitions only.

\subparagraph*{Acknowledgements}
This research was initiated when the first author visited
the Department of Computer Science
at TU Eindhoven during the winter 2015--2016.
He wishes to express his
gratitude to the other authors and the department for their
hospitality.

\bibliographystyle{plain}

\newcommand{\cgta}{\emph{Comput.\ Geom.\ Theory Appl.}}
\newcommand{\dcg}{\emph{Discr.\ Comput.\ Geom.}}
\newcommand{\ijcga}{\emph{Int.\ J.\ Comput.\ Geom.\ Appl.}}
\newcommand{\alg}{\emph{Algorithmica}}
\newcommand{\jalg}{\emph{J.\ Alg.}}
\newcommand{\ipl}{\emph{Inf.\ Proc.\ Lett.}}

\newcommand{\socg}[1]{In \emph{Proc.\ #1 ACM Symp.\ Comput.\ Geom.\ (SoCG)}}
\newcommand{\soda}[1]{In \emph{Proc.\ #1 ACM-SIAM Symp.\ Discr.\ Alg.\ (SODA)}}
\newcommand{\esa}[1]{In \emph{Proc.\ #1 Europ.\ Symp.\ Alg.\ (ESA)}}
\newcommand{\walcom}[1]{In \emph{Proc.\ #1 Int.\ Workshop Alg.\ Comput.\ (WALCOM)}}
\newcommand{\cccg}[1]{In \emph{Proc.\ #1 Canad.\ Conf.\ Comput.\ Geom.\ (CCCG)}}
\newcommand{\wads}[1]{In \emph{Proc.\ #1 Workshop Alg.\ Data Struct.\ (WADS)}}
\newcommand{\cats}[1]{In \emph{Proc.\ #1 Comput.: Australas.\ Theory Symp.\ (CATS)}}
\newcommand{\cocoon}[1]{In \emph{Proc.\ #1 Int.\ Comput.\ Comb.\ Conf.\ (COCOON)}}
\newcommand{\stoc}[1]{In \emph{Proc.\ #1 Annu.\ ACM SIGACT Symp.\ Theory Comput.\ (STOC)}}

\end{document}